%% file: JPhysA.tex
\documentclass[12pt]{iopart}

\usepackage{citesort}

\usepackage{iopams, bbm,graphicx, color} 
\expandafter\let\csname equation*\endcsname\relax 
\expandafter\let\csname endequation*\endcsname\relax 
\usepackage{amsmath, amsthm}
\usepackage[toc,page]{appendix}

\newcounter{counter}
\newtheorem{theorem}[counter]{Theorem}
\newtheorem{lemma}[counter]{Lemma}

\newtheorem{cor}[counter]{Corollary}
\newtheorem*{prob}{Problem Statement}
\newtheorem{example}{Example}

\newtheorem*{remark}{Remark}

\newcommand*{\defeq}{\mathrel{\vcenter{\baselineskip0.5ex \lineskiplimit0pt
                     \hbox{\scriptsize.}\hbox{\scriptsize.}}}%
                     =}

\makeatletter
\renewcommand*\env@matrix[1][c]{\hskip -\arraycolsep
  \let\@ifnextchar\new@ifnextchar
  \array{*\c@MaxMatrixCols #1}}
\makeatother

\newcommand{\BigO}[1]{\ensuremath{\operatorname{O}\left(#1\right)}}

\begin{document}

\title[Stability of the Trotter-Suzuki decomposition]{
 Stability of the Trotter-Suzuki decomposition}
\author{Ish Dhand and Barry C Sanders}
\address{Institute for Quantum Science and Technology, University of Calgary, Alberta T2N 1N4, Canada}
\address{Hefei National Laboratory for Physical Sciences at the Microscale,
	University of Science and Technology of China, China}
\eads{\mailto{idhand@ucalgary.ca} and \mailto{sandersb@ucalgary.ca}}

\begin{abstract}
The Trotter-Suzuki decomposition is an important tool for the simulation and control of physical systems. We provide evidence for the stability of the Trotter-Suzuki decomposition. We model the error in the decomposition and determine sufficiency conditions that guarantee the stability of this decomposition under this model.  We relate these sufficiency conditions to precision limitations of computing and control in both classical and quantum cases. Furthermore we show that bounded-error Trotter-Suzuki decomposition can be achieved by a suitable choice of machine precision.
\end{abstract}
%\vspace{2pc}â„
%\maketitle

%%%%%%%%%%%%%%%%%%%%%%%%%%%%%%
%%%%%%%%%%%%%%%%%%%%%%%%%%%%%%

\section{Introduction} 
\label{sec:Introduction}
Solving partial linear differential equations is essential in science.
Many systems obey equations of the form 
\begin{equation}\label{eq:v}
	\partial_\lambda {v}(\lambda) = A{v}(\lambda)
\end{equation}
with~${v}$ an element of an $\ell$-dimensional Banach space $X$, $A$ a bounded operator whose domain is in $X$ and $\lambda \in [0,\infty)$.The vector~$v$ has norm~$\left\|v\right\|$ determined uniquely by the norm property of the Banach space.  For example, time evolution under the Schr\"odinger equation has the form~\eqref{eq:v} if $A$ is skew-adjoint and $X$ is a complex Hilbert space. If $A$ is self-adjoint, the equation describes the cooling (or heating) of a physical system and is useful in finding the ground state of quantum systems.

Given constant $A\in X$, a solution to Eq.~\eqref{eq:v} is the continuous function ${v}: [0,\infty) \to X$, differentiable on $(0,\infty)$ that satisfies the equation. If the solutions of Eq.~\eqref{eq:v} are unique, then we can define the operator exponential $U(\lambda)$ as
\begin{equation}
	{v}(\lambda) = U(\lambda){v}(0) 
\end{equation}
for initial condition ${v}(0)$. Here $U(\lambda)$ an element of $L(X)$, which is the set of isomorphisms on $X$. Geometrically, $U(\lambda)$ is the flow generated by constant $A$
\begin{equation}
	A = \lim_{\lambda\to 0} {\lambda}^{-1} (U(\lambda) - \mathbbm{1}).
\end{equation}
In general, an expression for either $v(\lambda)$ or $U(\lambda)$ obtained by solving Eq.~\eqref{eq:v} is elusive. In such cases we have to approximate $v(\lambda)$ with some ${v}^\prime(\lambda)$
such that the distance between the exact and the approximate solutions
is smaller than a specified bound. 
Equivalently, we can approximate $U(\lambda)$ with ${U}^\prime(\lambda)$,
which is close to flow $U(\lambda)$ in terms of a suitable distance measure defined on the operator space. 
Please note that we use the prime in $v^\prime$ and $U^\prime$ not to represent derivative but to represent objects close to $v$ and $U$ respectively.

The distance between $v(\lambda)$ and ${v}^\prime(\lambda)$ is defined as the Banach-space norm of their difference $\left\|v(\lambda)-{v}^\prime(\lambda)\right\| $. Similarly, distance between two operators is defined in terms of their operator norms as $\left\|U(\lambda)-{U}^\prime(\lambda)\right\| $ for the operator norm induced by the Banach-space norm defined as
\begin{equation}\label{eq:OperatorNorm}
	\left\|W\right\|\defeq \max\left\{\left\| W v\right\| : \forall v \in X,  \|v\| \le 1 \right\}.
\end{equation}
For instance, if the distance between two states in a Banach space is given by the $2-$norm $\left\|v(\lambda)-{v}^\prime(\lambda)\right\|$, then the corresponding natural norm~\cite{Horn1985} for operators on the same space is the spectral norm, which we denote $\left\|U(\lambda) - {U}^\prime(\lambda)\right\|$. The $2$-norm for the vectors is $\left\|v\right\| = \sqrt{\sum_j v_j^2}$ for ${v_j}$ the elements of the column matrix representing the vector. The spectral norm is calculated as $\left\|U\right\| = \sqrt{\lambda_\text{max} \left(U^{\dagger} U\right)}$ where $\lambda_\text{max}: L(X)\to \mathbbm{R} $ is the magnitude of the largest eigenvalue of the respective operator.

One approach to approximating the flow $U(\lambda)$ uses product formul\ae, which has the advantage that the approximate decomposition preserves the geometric properties of the original flow. For example, a flow that is a unitary operator generated by skew-Hermitian~$A$ 
is decomposed into a product of unitary operators.
Product formul\ae~approximate the flow 
\begin{equation}
U(\lambda) =  \text{e}^{\sum_{j=1}^m A_j \lambda} 
\end{equation}
generated by $A=\sum_{j=1}^m A_j \lambda$ with
\begin{equation}
 	{U}^\prime(\lambda) =\prod_{p=1}^N U_p,
\end{equation}
which is a product of~$N$ ordinary operator exponentials 
\begin{equation}\label{eq:UpDefinition}
U_p \defeq {\text{e}^{A_{j_p}\lambda_p}}
\end{equation}
of $A_j$. Here $\{\lambda_p\}$ is a sequence of real numbers such that $\sum_p \lambda_p = \lambda$.

The recursive Trotter-Suzuki decomposition (TSD) is an especially important product formula, which plays a key role in classical~\cite{Bandrauk1991,Bandrauk1993,Vidal2004,Couairon2007,Jordan2007,Vidal2007,Verstraete2008,Banuls2009,Zhao2010} and quantum algorithms involving Hamiltonian simulation~\cite{abrams1997,Lloyd1996a,Harrow2009c,whitfield2011,Jordan2012a} and in quantum control~\cite{Schirmer1999}. The TSD is important as it minimizes the computational cost (quantified by number of ordinary operator exponentials in decomposition) of approximating ordered operator exponentials using product formulae~\cite{Berry2006,papageorgiou2012efficiency}. We focus our discussion on the TSD but our analysis can be used to study the stability properties of product formul\ae~such as the Baker-Campbell-Hausdorff formula~\cite{Blanes2009,Gilmore1974,Wilcox1967} and the Magnus expansion~\cite{Magnus1954,Richtmyer1965}. 

Recent quantum algorithms for Hamiltonian simulation~\cite{Berry2013,Berry2014} have a computational  cost that is exponentially smaller as function of desired error tolerance than algorithms based on product formul\ae. However, these algorithms do not perform ordered operator decomposition are not intended for use in quantum control or classical simulation. Hence, we do not consider these algorithms here.

The advantage of TSD is provable mathematically. However, real-world implementation of TSD as a calculation on a finite-precision Turing-equivalent computer (classical computer) or as an imperfect gate sequence on a quantum computer
or as a sequence of experimental control operations no longer guarantees this advantage. Hence, we have the following problem statement.\begin{flushright}

\begin{prob}
The stability of the TSD under finite precision either for computing or control is not proven.
Instability of the TSD could prevent achieving minimum time complexity and consequently makes some instances of computation or control infeasible.\end{prob}
\end{flushright}

%%%%%%%%%%%%%%%%%%%%%%%%%%%%%%
%%%%%%%%%%%%%%%%%%%%%%%%%%%%%%

\section{Background} 
The Trotter Suzuki decomposition has been studied since 1959 when Trotter~\cite{Trotter1959} considered semi-groups of operators acting on Banach space. Later, it was shown~\cite{Suzuki1976,Suzuki1976a}  that 
\begin{equation}
U(\lambda) = U^\prime_{T}(\lambda) + \BigO{\lambda^2/r},
\end{equation}
where $U_T(\lambda)$, the Trotter decomposition is given by 
\begin{equation}\label{eq:Trotter}
U^\prime_{T}(\lambda) = \left(\text{e}^{A_1 \lambda/r}\text{e}^{A_2 \lambda/r}\cdots \text{e}^{A_m \lambda/r} \right)^r.
\end{equation}
For product formul\ae~like~\eqref{eq:Trotter}, the number of exponentials of $A_j$ needed in the decomposition quantifies the computational cost of the decomposition~\cite{Berry2006}. In the case of the Trotter formula~\eqref{eq:Trotter}, this cost scales~\cite{Lloyd1996a}  as $\BigO{\lambda^2}$ and can be improved~\cite{Berry2006} to $\BigO{\lambda^{1+1/2k}}$ for any integer $k$ by using the $k^\text{th}$-order recursive TSD~\cite{Suzuki1990,Suzuki1991,Suzuki1992,Suzuki1992a,Suzuki1993,Suzuki1994,Bandrauk2013}.

The TSD is given by the recursion relation
\begin{align}
	U_S^{(2)}(\lambda) &= \prod_{j =1}^{m} \text{e}^{A_j\lambda/2}\prod_{j^\prime =m}^{1}\text{e}^{A_{j^\prime}\lambda/2},\nonumber \\
	U_S^{(2k)}(\lambda) &= [U_S^{(2k-2)}(p_k\lambda)]^2  U_S^{(2k-2)}((1-4p_k)\lambda) [U_S^{(2k-2)}(p_k\lambda)]^2,
	\label{eq:S2k}
\end{align} 
where~$\{1/4 \le p_k \le 1\}$ is some sequence of real numbers and $S$ denotes Suzuki. This recursion relation~\eqref{eq:S2k} comprises one backward step ($1-4p_k < 0$) and four forward steps ($p_k > 0$).
The choice of this sequence determines how well the TSD performs. 
If we regard that the TSD is being performed by an infinitely precise machine, then the  error in the TSD arises due to non-commutativity of the generators $\{A_i\}$ . This ideal error is defined in terms of the operator norms as
\begin{equation}
	\label{eq:IdealError}
	\varepsilon_\text{ideal}\defeq\frac{\left\|\text{e}^{\sum_i A_i \lambda} - U_S^{(2k)}(\lambda)\right\|}{\left\|\text{e}^{\sum_i A_i \lambda} \right\|}
	= \BigO{\lambda^{1+1/2k}}
\end{equation}
%where we define
%\begin{equation}\label{eq:epsilonDef}
%\end{equation}
and is minimized by choosing~\cite{Suzuki1990}
\begin{equation}
\label{eq:p_k}
	p_k = \frac{1}{{4-4^{\frac{1}{2k-1}}}}.
\end{equation}

Equation~\eqref{eq:S2k} has one backward and four forward steps; a different number of steps could be used but such alternatives are avoided for the following reasons. Suzuki has shown~\cite{Suzuki1995} that a fully-positive (no backwards steps) or asymmetric (under reversal of order in which the operators act) decomposition lead to error scaling worse than that given by Eq.~\eqref{eq:IdealError}. A decomposition with one backward and two forward steps is unstable because, in this case, the coefficients $p_k$ diverge~\cite{Suzuki1995} with increasing $k$. We can easily see that a TSD with more than four forward steps has the same error scaling~\eqref{eq:IdealError}, but the constant coefficients in the scaling~\eqref{eq:IdealError} are then sub-optimal. Any other symmetric decomposition of $U(\lambda)$ can be reduced to these cases by combining sequences of unitary operators into a single unitary operators, and the case of one backwards and four forward~\eqref{eq:S2k} is optimal

The ideal error~\eqref{eq:IdealError} is minimized by dividing $\lambda$ into $r$ intervals before performing TSD of $\exp(A\lambda/r)$ to a finite order $k$ and concatenating these $r$ intervals together~\cite{Berry2006}.
This error is known to converge with the order $k$ of TSD~\cite{Suzuki1993,Suzuki1994}.

%%%%%%%%%%%%%%%%%%%%%%%%%%%%%%
%%%%%%%%%%%%%%%%%%%%%%%%%%%%%%

\section{Error Model}
\label{sec:ErrorModel}

Experimental or computational approximation of a unitary operation is imperfect. We show that we can model this imperfect approximation with unitary matrices plus a small matrix-valued~\cite{Tao2012} Gaussian random error. We argue that this model is valid for experimental implementation and classical computation and fair in the case of quantum computation.

The limitation of the experimental accuracy and the round-off error on classical or quantum computers can be captured by the same quantity, which we refer to as the machine error:
\begin{align}
 \varepsilon & \defeq \frac{\left\| U_S^{(2k)}  - \tilde{U}_S^{(2k)}\right\|}{{\left\| U_S^{(2k)}\right\|}} \\
&= \frac{\left\| \prod_{p=1}^N{U_p}  - \prod_{p=1}^N\tilde{U}_p\right\|}{{\left\|\prod_{p=1}^N{U_p} \right\|}}.
\label{eq:error}
\end{align} 
Above and henceforth, we use $\tilde{\bullet}$ to represent machine approximation of $\bullet$, which is an operator on the Banach space. The incurred net error 
\begin{align}
\varepsilon_\text{net} &\defeq \frac{\left\|\text{e}^{\sum_i A_i \lambda } - \prod_{p=1}^N\tilde{U}_p \right\|}{{\left\|\text{e}^{\sum_i A_i \lambda}\right\|}}\nonumber\\
 &\le \frac{\left\|\text{e}^{\sum_i A_i \lambda } - \prod_{p=1}^N{U_p} \right\|}{{\left\|\text{e}^{\sum_i A_i \lambda}\right\|}} + \frac{\left\| \prod_{p=1}^N{U_p}  - \prod_{p=1}^N\tilde{U}_p\right\|}{{\left\|\prod_{p=1}^N{U_p} \right\|}}
\label{eq:Both}
\end{align}
arises from two sources corresponding to the two terms in~\eqref{eq:Both}. The first term in~\eqref{eq:Both} is the error in the TSD assuming perfectly precise operators. This error arises from the non-commutativity of the implemented operators in the quantum and experimental cases and of the matrix multiplications performed in the case of classical computation. The second term is a consequence of the machine error~\eqref{eq:error}. Figure~\ref{fig:Error} depicts the two sources of error.
\begin{figure}[h]
  \centering
  \def\svgwidth{200pt}
  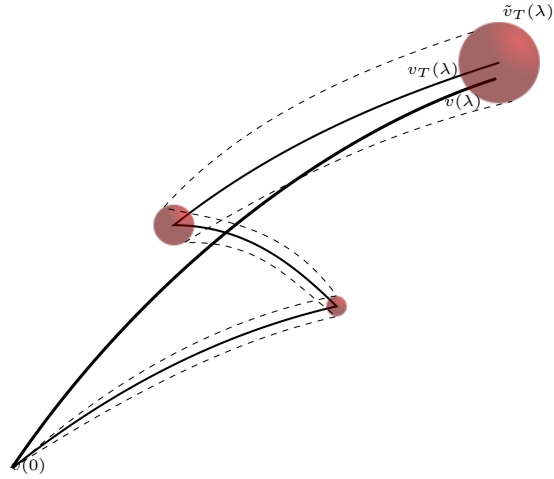
  \caption{A simple representation of the concept of growing machine error in the TSD.  For given~$U(\lambda)$, we see two trajectories from~$v(0)$. One trajectory follows a smooth line to~$v(\lambda)$ and represents the ideal transformation imposed by~$U(\lambda)$. The second trajectory follows three sequential smooth lines to~$v_\text{T}(\lambda)$, which is close to~$v(\lambda)$ but not equal. Each of these sequential smooth lines in the second case represents one~$U_p(\lambda)$ in the decomposition $U(\lambda) = \prod_{p=1}^NU_p$. The spheres in the figure repesent the set of states that the trajectory $\tilde{v}_T(\lambda)$ can attain in the presence of machine error.} \label{fig:Error}
\end{figure}

Machine error~$\varepsilon$~\eqref{eq:error} arises in quantum computers because a finite sequence of quantum gates would be used to approximate desired gates from the set $\{U_p\}$.  This approximation is performed using efficient algorithms~\cite{Kitaev2002classical,Harrow2002,Dawson2005,Kliuchnikov2013},  which have time complexity poly-logarithmic in $1/\varepsilon$. In the case of classical computers, which utilize floating point precision to perform the TSD, we use $\varepsilon$ to treat the imperfection due to round-off error. In experiments, machine error $\varepsilon$ arises from limited control precision.

To investigate the error $\varepsilon$ in $U_s^{(2k)}$, we first consider the error incurred in the ordinary operator exponentials $U_p$, which are imperfectly implemented as $\tilde{U}_p$ in computation or experiment. In experiments, repeated runs will encounter different error and hence, $\tilde{U}_p$ can be treated as random variables. Computationally, the probability density function (pdf) of $\left(\tilde{U}_p\right)$ will be a delta function if the TSD is performed repeatedly on the same computer; i.e., $\tilde{U}_p$ will not be correct and the discrepancy will be constant. Implementing the TSD repeatedly on machines with different instruction sets will lead to a non-zero spread in the pdf of $\left(\tilde{U}_p\right)_{ij}$.

We examine the distribution of the implemented operators $\tilde{U}_p$ in terms of the individual elements $\left(\tilde{U}_p\right)_{ij}$ in the matrix representation of $\tilde{U}_p$. Specifically, we treat these matrix elements as random variables and analyze their pdf.

The elements $\left(\tilde{U}_p\right)_{ij}$ of the matrix representations of the TSD operators depend upon a large number of other random variables. In computation or experiment, these latter random variables arise in intermediate steps of computation of $U_p$. Hence, the central limit theorem motivates the assumption that the matrix elements~$\left(\tilde{U}_p\right)_{ij}$ are Gaussian random variables.

Floating-point error in classical computation is proportional to the magnitude of the element. For mathematical tractability, we make the reasonable assumption that error in experiment and quantum computation is proportional to the magnitude of the implemented operation. Hence, we assume that the standard deviation $\sigma\left(\left(\tilde{U}_p\right)_{ij} \right)$ of the pdf for each matrix element is proportional to the size of the element itself:
\begin{equation}
	\sigma\left(\left(\tilde{U}_p\right)_{ij} \right) = \epsilon_m\left|\left(U_p\right)_{ij}\right|.
	\label{eq:elementError2}
\end{equation}
In other words, we define the constant machine epsilon $\epsilon_m$ as the relative, and not absolute, standard error in the elements of the matrix representation. This machine epsilon is a property of the computer or experiment used to effect the TSD and is fixed for a given implementation. To summarize, we treat these matrix elements as independent Gaussian random variables 
\begin{equation}
\label{eq:elementError}
\left(\tilde{U}_p\right)_{ij} \sim \mathcal{N}\left(\left(U_p\right)_{ij}, \epsilon_m^2\left|\left(U_p\right)_{ij}\right|^2\right) \forall~i,j \in \{1,2,\cdots,\ell\}~\forall p \in{1,2,\cdots,N}.
\end{equation}
with mean equal to the corresponding elements of $U_p$ and relative standard deviation $\epsilon_m$.

Our error model is only approximately valid for quantum computing as the unitarity of implemented operators contradicts the independence of the random matrix elements. We consider this simple model because it is amenable to rigorous theoretical analysis. This analysis lays a foundation for investigating the stability of product formul\ae~against more sophisticated error models.   

We employ the pdf~\eqref{eq:elementError} of $\left(\tilde{U}_p\right)_{ij}$ to analyze the pdf of $\varepsilon$ and consequently provide evidence for the stability of the TSD. The pdf
\begin{equation}
	P\left(\varepsilon; \left\{\left(U_p\right)_{ij}\right\}, \epsilon_m\right),
\end{equation}
of $\varepsilon$ is parametrized by the elements $\left\{\left(U_p\right)_{ij}\right\}$ of the ideal TSD operators and the machine epsilon $\epsilon_m$. In Section~\ref{sec:Theorems}, we use the mean $\mu(\varepsilon) \defeq E(\varepsilon)$ and standard deviation $\sigma(\varepsilon) \defeq \sqrt{E\left([\varepsilon-\mu(\varepsilon)]^2\right)}$ of this pdf $P(\varepsilon)$ to characterize the stability of the TSD for $E(\bullet)$ denoting expectation value.

We now present an example to illustrate the properties of the pdf corresponding to a simple case of operators $U_p$. We show that even in the simple case of  scalar multiplication, i.e.~$\ell = 1$ real matrices, $\mu(\varepsilon)$ and $\sigma(\varepsilon)$ diverge exponentially with $N$ and indicate instability in the straightforward implementation of the TSD. In the next section, we show that this instability in the TSD can be removed by suitably normalizing the operator elements.

\begin{example}\label{ex:EllOne}
If the operators $\{U_p\}$  are scalar, i.e., $1\times 1$ matrices with positive real number entries $\{r_p\}$ respectively, then the mean~$\mu(\varepsilon)$ and standard deviation~$\sigma(\varepsilon)$ of error $\varepsilon$~\eqref{eq:error} diverge exponentially with $N$ as  
\begin{align}
\mu\left(\varepsilon\right) &\geq \exp\left({N\epsilon_m^2/2}\right)-1,\\ 
\sigma(\varepsilon) &\geq \left(\exp\left(2N\epsilon_m^2\right)- \exp\left(N\epsilon_m^2\right) - 2\exp\left(N\epsilon_m^2/2\right)  +1\right)^{{1}/{2}}.
\end{align}
\end{example}
\begin{proof}[Solution]
The error due to machine precision $\varepsilon$ is given by 
\begin{align}
\varepsilon = \frac{\left\| \prod_{p=1}^N{r_p}  - \prod_{p=1}^N\tilde{r}_p\right\|}{{\left\|\prod_{p=1}^N{r_p} \right\|}} = \left| 1 - \prod_{p=1}^N X_p \right|,
\label{eq:varepsilonExample}
\end{align}
with 
\begin{equation}
X_p \defeq \frac{\tilde{r}_p}{ | r_p |}
\end{equation}
for $|\bullet|$ the absolute value of the real-number argument. Here $\left\{X_p\right\}$ is a set of independent and identically distributed Gaussian random variables 
\begin{equation}\label{eq:GaussianRandom}
X_p \sim \mathcal{N}\left(1,\epsilon_m^2\right)	.
\end{equation}

We first analyze the distribution of the product of the Gaussian random variables \eqref{eq:GaussianRandom} and use the analysis to learn about the distribution of $\varepsilon$. The product 
\begin{equation}
X \defeq \prod_{p=1}^N X_p
\end{equation}
is described by the log-normal distribution~\cite{Johnson1997} because 
\begin{equation}\label{eq:LogSumProduct}
\log(X) = \log\left(\prod_{p=1}^N X_p\right) = \sum_{p=N}^{1} \log\left(X_p\right),
\end{equation}
which is a sum of many independent random variables. For large $N$, this sum~\eqref{eq:LogSumProduct} follows the Gaussian distribution from the central limit theorem. The pdf  corresponding to the log-normal distribution of $X$ is
\begin{equation}
\label{eq:log-normal}
P_{\text{log-normal}}\left(X; \left\{r_p\right\}, \epsilon_m\right) = \frac{1}{X\sqrt{2\pi N}\epsilon_m} \exp\left(-\frac{\left(\ln X \right)^2}{2N\epsilon_m^2}\right),
\end{equation}
and its mean and and standard deviation
\begin{align}
\mu(X) &= \exp\left({N\epsilon_m^2/2}\right),\label{eq:MeanProduct}\\
\sigma(X) &= \left(\exp\left(2N\epsilon_m^2\right)- \exp\left(N\epsilon_m^2\right)\right)^{{1}/{2}} \label{eq:StdProduct}
\end{align}
 diverge exponentially~\cite{Johnson1997} with $N$:

We now use our knowledge of the distribution of $X$ to find the pdf of $\varepsilon=|1-X|$ in terms of the pdf of the log-normal distribution
\begin{align}
P(\varepsilon = |1-X|; \epsilon_m) = 
\left\{\begin{matrix}[l]
0, & \qquad X \le 0,\\
P_\text{log-normal}(1-X) + P_\text{log-normal}(1+x), & \qquad 0 < X.
\end{matrix}\right.
\end{align}
and, because $P_\text{log-normal}(\bullet)$ is defined only for positive values of $\bullet$, we have
\begin{align}
P(\varepsilon = |1-X|; \epsilon_m) = 
\left\{\begin{matrix}[l]
0, ~&\qquad \phantom{0<}~X \le 0,\\
P_\text{log-normal}(1-X) + P_\text{log-normal}(1+X), ~& \qquad 0 < X \le 1, \\
P_\text{log-normal}(1+X), ~& \qquad 1 < X. 
\end{matrix}\right.
\end{align}
Thus, $\varepsilon$ is described by the following folded~\cite{Leone1961} log-normal distribution
\begin{align}
P\left(\varepsilon;  \epsilon_m\right) =  
\left\{\begin{matrix}[l]
0, &\qquad\phantom{0<}~\varepsilon \le 0,\\
\frac{1}{\sqrt{2\pi N}\epsilon_m}\left[\frac{1}{(1-\varepsilon)}\exp\left(-\frac{\left(\ln (1-\varepsilon) \right)^2}{2N\epsilon_m^2}\right) +  
\frac{1}{(1+\varepsilon)} \exp\left(-\frac{\left(\ln (1+\varepsilon) \right)^2}{2N\epsilon_m^2}\right)\right], &\qquad 0 \le \varepsilon \le 1,\\
\frac{1}{\sqrt{2\pi N}\epsilon_m}\left[\frac{1}{(1+\varepsilon)} \exp\left(-\frac{\left(\ln (1+\varepsilon) \right)^2}{2N\epsilon_m^2}\right)\right], &\qquad 1 < \varepsilon.	
\end{matrix}\right.
\end{align}

The mean and standard deviation of $X$~\eqref{eq:log-normal} are straightforward but those of $\varepsilon$ are not. Consequentially, we only obtain bounds
\begin{align}
\mu\left(\varepsilon\right) &\geq \exp\left({N\epsilon_m^2/2}\right)-1,\label{eq:MeanLog-normal}\\ \sigma(\varepsilon) &\geq \left(\exp\left(2N\epsilon_m^2\right)- \exp\left(N\epsilon_m^2\right) - 2\exp\left(N\epsilon_m^2/2\right)  +1\right)^{{1}/{2}}.\label{eq:StdLog-normal}
\end{align}
In obtaining Eqs.~\eqref{eq:MeanLog-normal} and~\eqref{eq:StdLog-normal}, we have used
\begin{align}
\mu\left(|1-X|\right) \ge \mu\left(X-1\right) = \mu(X) - 1
\end{align}
and
\begin{align}
\sigma^2(|1-X|) &= \mu\left(|X -1|^2\right) - \mu^2(|X-1|) \nonumber\\
&= \sigma^2(X) + \mu\left(|X -1|^2\right) -\mu\left(X^2\right) +\mu^2\left(X\right) - \mu^2(|X-1|) \nonumber\\
&= \sigma^2(X) - \mu\left(2X-1\right) +\mu\left(X\right)^2 - \mu(|X-1|)^2 \nonumber\\
&\ge \sigma^2(X) - \mu\left(2X-1\right) =  \sigma^2(X) - 2\mu(X) + 1
\end{align}
along with values of $\mu(X)$ and $\sigma(X)$ from Eqs.~\eqref{eq:MeanProduct} and~\eqref{eq:StdProduct}.
Even in the simplest case of $\ell = 1$ decomposition with real operator elements, the mean and the standard deviation of the $\varepsilon$-distribution diverge exponentially with respect to $N$ as shown by the lower bound~\eqref{eq:MeanLog-normal} and~\eqref{eq:StdLog-normal}. This exponential divergence is an indication of the instability in the straightforward implementation of the TSD. 
\end{proof} 

In this section we have introduced a model for error in the TSD. We defined the distribution of the error and the mean and standard deviation of the same distribution. We also showed that for the case of $\ell = 1$ and positive real matrix elements, the mean and standard deviation of this distribution diverge, thereby indicating instability in the TSD. In the following section, we use the error distribution model introduced above to describe the stability properties of a general TSD and show that the TSD can be made stable by ensuring that the operator norms are bounded as described in the next section.

%%%%%%%%%%%%%%%%%%%%%%%%%%%%%%
%%%%%%%%%%%%%%%%%%%%%%%%%%%%%%

\section{Instability and stabilization}
\label{sec:Theorems}
In this section we show that the TSD is unstable in the presence of machine error~\eqref{eq:elementError}. This instability results from the standard deviation $\sigma(\varepsilon)$ of machine error $\varepsilon$ increasing exponentially with $N$ and implies that the implemented TSD might differ exponentially from the expected decomposition. 

On the other hand, if the values of both the standard and mean error grow polynomially in $N$, then there is exponentially (in $N$) small probability of $\varepsilon$ being exponentially large. In this latter case, the implemented TSD operator $\tilde{U}_S^{(2k)}$  provides a good approximation to the actual decomposition $U_S^{(2k)}$, and the decomposition is stable. We also show that the TSD can be stabilized by suitably setting the norm of operators in the TSD.

In the following propositions and the accompanying proofs, we use the  spectral norm on operator space as a distance measure because of the importance of this norm to quantum systems as described in Section~\ref{sec:Introduction}. Our results can be easily modified to other descriptions of error. 

In Lemma~\ref{lem:normEpsilon}, we relate the standard error of the TSD operators $\{U_p(\lambda_p)\}$~\eqref{eq:UpDefinition} to the machine epsilon~$\epsilon_m$~\eqref{eq:elementError2}. As expected, the upper and lower bound of the error in the TSD operators turn out to be proportional to the machine epsilon.

\begin{lemma} 
\label{lem:normEpsilon} Given machine epsilon $\epsilon_m$ and linear operators $\{A_i\}$, which act on $\ell$-dimensional Banach space and $A = \sum_j^m A_j$, the relative standard error of each $U_p$ is bounded above and below as
\begin{equation}
\epsilon_m \leq \sigma\left(\frac{\left\|U_p - \tilde{U}_p\right\|}{\left\|U_p\right\|}\right) \leq \epsilon_m \sqrt{\ell} \quad \forall p. \label{eq:Lem1}
\end{equation}
\end{lemma}
\begin{proof}
The lower bound is saturated by a matrix with only one non-zero element. The lower-bound equality holds trivially in the special case of $\ell = 1$, as described in Example~\ref{ex:EllOne}.

The upper bound of the spectral norm follows from the spectral distribution~\cite{Bai1999}, which is the distribution of the singular values of the matrix, of the normalized random matrix $\frac{\left\|U_p - \tilde{U}_p\right\|}{\left\|U_p\right\|}$. The distribution of the square of singular values tends to the following limiting pdf
\begin{align}
p\left(x_p^2 \defeq\left[\frac{\left\|U_p - \tilde{U}_p\right\|}{\left\|U_p\right\|}\right]^2\right) = 
\left\{
\begin{array}{ll}
\frac{\sqrt{(4\epsilon_m^2\ell - x_p^2)(x_p^2)}}{2\pi x_p^2\epsilon_m^2\ell},&~\text{if}~ 0\le x_p \le 2\epsilon_m\sqrt{\ell},\\
0,& ~\text{otherwise}.
\end{array}\right.\label{eq:pdfBai}
\end{align}
The probability density is non-zero only in the domain $[0,2\epsilon_m\sqrt{l}]$. Thus, its standard deviation~\eqref{eq:Lem1} bounded above by $\epsilon_m \sqrt{\ell}$.
\end{proof}

In Lemma~\ref{lem:normEpsilon}, we proved upper and lower bounds of the standard deviation of the error in a single operator $U_p$ under the assumption that the error is of the form~\eqref{eq:elementError}. The TSD $U_S^{(2k)}$ comprises $\BigO{\tau^{1+1/2k}}$ operators of the form $U_p$. In Theorem~\ref{thm:Unstable}, we use the bounds derived in Lemma~\ref{lem:normEpsilon} to relate the TSD error $\varepsilon$ to the machine epsilon. Specifically, we calculate a lower bound for $\sigma(\varepsilon)$ for a straightforward implementation of the TSD, in which the operators $U_p$ are not normalized before composition. The lower bound is an exponentially growing function of $N$, the number of operators in the decomposition. This exponential growth in error implies the instability of the straightforward implementation of the TSD.

\begin{theorem}
\label{thm:Unstable}
Given machine epsilon $\epsilon_\text{m}$ and  linear operators ${A_i}$, which act on $\ell$-dimensional Banach space and $A = \sum_{i=1}^m A_i$, the standard deviation 
\begin{equation}
\sigma\left(\varepsilon\right) \ge  N\ell^{(N-1)/2}\epsilon_m .
\end{equation}
of error of the TSD is exponentially divergent with respect to $N$.
\end{theorem}

\begin{proof} 
In this proof, we use the matrix representation of the TSD operators $\tilde{U}_S^{(2k)}(\lambda) = \prod_{p = 1}^N \tilde{U}_p(\lambda_p)$. Consider the matrix elements of the matrix product
\begin{align}
	\left(\tilde{U}_S^{(2k)}\right)_{i_Ni_0} &= \sum_{i_{N-1}=1}^\ell\dots\sum_{i_2=1}^\ell\sum_{i_1=1}^\ell \left[\left(\tilde{U}_N\right)_{i_N,i_{N-1}}\dots\left(\tilde{U}_3\right)_{i_3,i_2}\left(\tilde{U}_2\right)_{i_2,i_1}\left(\tilde{U}_1\right)_{i_1,i_0}\right]\nonumber\\
	&= \sum_{i_{N-1}=1}^\ell\dots\sum_{i_2=1}^\ell\sum_{i_1=1}^\ell \left[\prod_{p=1}^N(\tilde{U}_p)_{i_p i_{p-1}}\right]
	\label{elementTSD}
\end{align}
where each of the elements $(\tilde{U}_p)_{i_p i_{p-1}}$ on the right-hand side is a random variable with distribution described by Eq.~$\eqref{eq:elementError}$. The right side of Eq.~\eqref{elementTSD} consists of an $\ell^{N-1}$ term summation of products, each with $N$ such elements as factors. If the standard deviation of error in each element $\epsilon_m$~\eqref{eq:elementError} is small, then the standard deviation of error the product $\prod_{p=1}^N(\tilde{U}_p)_{i_p i_{p-1}}$ is the sum of the relative errors of the factors~\cite{Ku1966}. Hence, for the product of $N$ elements, which occurs in Eq.~\eqref{elementTSD}, we have
\begin{equation}
	\sigma\left(\prod_{p=1}^N\left(\tilde{U}_p\right)_{i_p i_{p-1}}\right) = N\epsilon.
\end{equation}

We consider the standard deviation of the error in the summation over $\ell^{N-1}$ terms of Eq.~\eqref{elementTSD}. From the central limit theorem, we conclude that this error we have Gaussian distribution with a standard deviation
\begin{align}
\sigma \left(\sum_{i_{N-1}=1}^\ell\dots\sum_{i_2=1}^\ell\sum_{i_1=1}^\ell \left[\prod_{p=1}^N(\tilde{U}_p)_{i_p i_{p-1}}\right]\right) &= \ell^{(N-1)/2} \sigma\left(\prod_{p=1}^N\left(\tilde{U}_p\right)_{i_p i_{p-1}}\right)\nonumber\\
&= \ell^{(N-1)/2}\left(N\epsilon_m\right).
\end{align}
From Lemma~\ref{lem:normEpsilon}, we know that the standard error in the norm of the $\tilde{U}_S^{(2k)}$ is always greater than the standard error of the elements of its matrix representation. We thus have 
\begin{equation}
\label{eq:Th2Error}
\sigma\left(\varepsilon\right) \ge  N\ell^{(N-1)/2}\epsilon_m;
\end{equation}
i.e., the standard error of the TSD diverges at least exponentially, as $\ell^{(N-1)/2}$, in the number $N$ of the operators in the TSD.
\end{proof}

In Theorem~\ref{thm:Unstable} above, we showed that the TSD is inherently unstable because $\sigma(\varepsilon)$ diverges. However, as we shall show in Theorem~\ref{thm:normalize}, this instability can be removed by ensuring correct normalization in the TSD. In Theorem~\ref{thm:normalize}, we provide a sufficiency condition on the norm of the TSD operators to ensure that $\sigma(\varepsilon)$ grows no faster than a linear function of $N$.
\begin{theorem}
\label{thm:normalize}
Given machine epsilon $\epsilon_\text{m}$ and linear operators $\{A_i\}$ which act on $\ell$-dimensional Banach space and $A = \sum_{i=1}^m A_i$,  if 
\begin{equation}
\sigma\left(\frac{\left\|\tilde{U}_p\right\|}{\left\|U_p\right\|}\right) \le \frac{1}{\sqrt{N}},
\label{eq:Condition}
\end{equation}
then the relative standard error
\begin{equation}
\sigma(\varepsilon) \le N\epsilon_m\sqrt{5\text{e}^2-4\rme}
\end{equation}
 in the TSD grows linearly in the worst case over the parameters $\left\{\left(U_p\right)_{ij}\right\}$ of the distribution.
\end{theorem}	

\begin{remark} Note that $\rme$ is the base of the natural logarithm.\end{remark}

\begin{proof}
		We are interested in obtaining an upper bound on the error
	\begin{align}
			\varepsilon =& \frac{\left\|U_{N} U_{N-1} \dots U_{2} U_{1} - \tilde{U}_{N} \tilde{U}_{N-1} \dots \tilde{U}_{2} \tilde{U}_{1}\right\|}{\left\|U_{N} U_{N-1} \dots U_{2} U_{1}\right\|}\nonumber\\
			\le&  \frac{\left\|U_{N} U_{N-1} \dots U_{2} U_{1} - U_{N} U_{N-1} \dots U_{2} \tilde{U}_{1}\right\|}{\left\|U_{N} U_{N-1} \dots U_{2} U_{1}\right\|}\nonumber\\
 &+ \frac{\left\|U_{N} U_{N-1} \dots U_{2} \tilde{U}_{1} - U_{N} U_{N-1} \dots \tilde{U}_{2} \tilde{U}_{1}\right\|}{\left\|U_{N} U_{N-1} \dots U_{2} U_{1}\right\|}\nonumber\\
 &+\cdots\nonumber\\
 &+ \frac{\left\|U_{N} U_{N-1} \dots \tilde{U}_{2} \tilde{U}_{1} - U_{N} \tilde{U}_{N-1} \dots \tilde{U}_{2} \tilde{U}_{1}\right\|}{\left\|U_{N} U_{N-1} \dots U_{2} U_{1}\right\|}\nonumber\\
 &+ \frac{\left\|U_{N} \tilde{U}_{N-1} \dots \tilde{U}_{2} \tilde{U}_{1} - \tilde{U}_{N} \tilde{U}_{N-1} \dots \tilde{U}_{2}\tilde{U}_{1} \right\|}{\left\|U_{N} U_{N-1} \dots U_{2} U_{1}\right\|},
\end{align}
where we have added and subtracted terms of the form $U_N\cdots U_{p+1}\tilde{U_p}\cdots\tilde{U_1}$ and have used the triangle inequality of the operator norm on the Banach space. This gives us
\begin{align}
\varepsilon \le&  \frac{\left\|U_{1} - \tilde{U}_{1}\right\|}{\left\|U_{1}\right\|} +
\frac{\left\|U_{2} - \tilde{U}_{2}\right\|}{\left\|U_{2}\right\|}\frac{\left\|\tilde{U}_1\right\|}{\left\|U_1\right\|} +\cdots \nonumber
\\
 &+ \frac{\left\|U_{N-1} - \tilde{U}_{N-1}\right\|}{\left\|U_{N-1}\right\|}\frac{\left\|\tilde{U}_{N-2}\dots\tilde{U}_{2}\tilde{U}_{1}\right\|}{\left\|U_{N-2}\dots U_{2}U_1\right\|}\nonumber\\ 
 &+\frac{\left\|U_{N} - \tilde{U}_{N}\right\|}{\left\|U_{N}\right\|}\frac{\left\|\tilde{U}_{N-1}\dots\tilde{U}_{2}\tilde{U}_{1}\right\|}{\left\|U_{N-1}\dots U_{2}U_1\right\|}.\label{eq:AfterTriangle}
	\end{align}
The terms in the above summation depend on random variables $\left(\tilde{U}_p\right)_{i_{p}i_{p-1}}$~\eqref{eq:elementError}). To obtain the upper bound of the standard error in this $N$-term summation, we replace each term by the one with the highest standard error, i.e., the final term~\eqref{eq:AfterTriangle} and thereby obtain the following upper bound on the standard deviation of the summation~\cite{Ku1966}
\begin{align}\label{eq:varepsilonLeq}
	\sigma^2\left(\varepsilon\right) &\le \sigma^2\left(\frac{\left\|U_p-\tilde{U}_p\right\|}{\left\|U_p\right\|}\frac{\left\|\tilde{U}_{N-1}\dots\tilde{U}_{2}\tilde{U}_{1}\right\|}{\left\|U_{N-1}\dots U_{2}U_1\right\|}\right)  N .
\end{align}

To calculate the standard deviation of the above product of random variables
\begin{equation}
X = \frac{\left\|U_p-\tilde{U}_p\right\|}{\left\|U_p\right\|} \hspace{1cm} Y = \frac{\left\|\tilde{U}_{N-1}\dots\tilde{U}_{2}\tilde{U}_{1}\right\|}{\left\|U_{N-1}\dots U_{2}U_1\right\|}
\end{equation}
 we use the identity%~\cite{goodman1960}
\begin{equation}
\sigma^2(XY) = \sigma^2(X)\sigma^2(Y) + \sigma^2(X)E^2(Y) + \sigma^2(Y)E^2(X), 
\end{equation}
where $\sigma^2(\bullet)$ and $E(\bullet)$ represent the variance and the expectation value respectively. Thus,
\begin{align}
\sigma^2\left(\frac{\left\|U_p-\tilde{U}_p\right\|}{\left\|U_p\right\|}\frac{\left\|\tilde{U}_{N-1}\dots\tilde{U}_{2}\tilde{U}_{1}\right\|}{\left\|U_{N-1}\dots U_{2}U_1\right\|}\right) 
\le~& \sigma^2\left(\frac{\left\|U_p-\tilde{U}_p\right\|}{\left\|U_p\right\|}\right)\sigma^2\left(\prod_{p=1}^N\frac{\left\| {\tilde{U}_p}\right\|}{\left\|{U_p} \right\|}\right) \nonumber\\
&+ \sigma^2\left(\frac{\left\|U_p-\tilde{U}_p\right\|}{\left\|U_p\right\|}\right)E^2\left(\frac{\left\|\tilde{U}_{N-1}\dots\tilde{U}_{2}\tilde{U}_{1}\right\|}{\left\|U_{N-1}\dots U_{2}U_1\right\|}\right)\nonumber \\
&+ \sigma^2\left(\prod_{p=1}^N\frac{\left\| {\tilde{U}_p}\right\|}{\left\|{U_p} \right\|}\right)E^2\left(\frac{\left\|U_p-\tilde{U}_p\right\|}{\left\|U_p\right\|}\right).\label{eq:identity}
\end{align}

We now evaluate the  standard error and the expectation values of the two relevant random variables in Eq.~\eqref{eq:identity}. We use Lemma $\ref{lem:normEpsilon}$ to obtain the standard error
\begin{equation}	\label{eq:MaxError}
\sigma\left(\frac{\left\|U_p-\tilde{U}_p\right\|}{\left\|U_p\right\|}\right) \le \epsilon_m \sqrt{\ell}.
\end{equation}
of $X$. The expectation value \begin{equation}
	E\left(\frac{\left\|U_p-\tilde{U}_p\right\|}{\left\|U_p\right\|}\right) \le 2\epsilon_m\sqrt{\ell}.
\end{equation}
of $X$ follows from the pdf of Eq.~\eqref{eq:pdfBai}.
 The random variable  $Y = \prod_{p=1}^N{\left\| {\tilde{U}_p}\right\|}/{\left\|{U_p} \right\|}$ is a product of $N$ independent positive random variables. Thus, the product has a log-normal distribution~\cite{Johnson1997}, which gives us
\begin{equation}
\sigma\left(\prod_{p=1}^N\frac{\left\| {\tilde{U}_p}\right\|}{\left\|{U_p} \right\|}\right) \le  
\sqrt{\rme^{2N\sigma^2\left(\frac{\left\| {\tilde{U}_q}\right\|}{\left\|{U_q} \right\|}\right)} - \rme^{N\sigma^2\left(\frac{\left\| {\tilde{U}_q}\right\|}{\left\|{U_q} \right\|}\right)}}
\end{equation}
and
\begin{equation}
	E\left(\prod_{p=1}^N\frac{\left\| {\tilde{U}_p}\right\|}{\left\|{U_p} \right\|}\right)  \le \rme^{\frac{N}{2}\sigma^2\left(\frac{\left\|\tilde{U}_p\right\|}{\left\|U_p\right\|}\right)}.
\end{equation}
These upper bounds on the mean and variance of the respective random variable correspond to the worst-case parameterization of $\varepsilon$.

Thus, under the condition that 
		\begin{equation}
		\sigma\left(\frac{\left\|\tilde{U}_p\right\|}{\left\|U_p\right\|}\right) \le \frac{1}{\sqrt{N}},
		\end{equation} we have
		\begin{align}
			\sigma\left(\varepsilon\right)
			&\le N\sqrt{\epsilon_m^2\ell\left(\rme^2-\rme\right) + \epsilon_m^2\ell\rme + 4\epsilon_m^2\ell\left(\rme^2-\rme\right)}\nonumber\\
			&= N\epsilon_m\sqrt{5\rme^2-4\rme}. \label{eq:final}
		\end{align}
Hence, the upper bound of $\sigma(\varepsilon)$ grows at most linearly in $N$ if the sufficiency condition~\eqref{eq:Condition} is obeyed.
\end{proof}

In Theorem~\ref{thm:normalize}, we showed that the error in the TSD grows linearly in the worst-case subject to the condition~\eqref{eq:Condition}. In fact, as we show in Corollary~\ref{cor:Classical}, it is possible to implement the TSD such that $\varepsilon$ is bounded by a prespecified error tolerance by setting the machine epsilon to be appropriately low. 
\begin{cor}
\label{cor:Classical}
Given error tolerance $\varepsilon_t$, if conditions
		\begin{equation}
			\sigma\left(\frac{\left\|\tilde{U}_p\right\|}{\left\|U_p\right\|}\right) \le \frac{1}{\sqrt{N}}\label{eq:Condition1},
		\end{equation}
 and
		\begin{align}
			\epsilon_m \le \frac{\varepsilon_t}{N\sqrt{\ell(5\rme^2-4\rme)}}\label{eq:Condition2}
		\end{align}
hold, then 
\begin{equation}
	\sigma(\varepsilon) \le  \varepsilon_t.
\end{equation}
\end{cor}
\begin{proof} 
The objective is to bound $\sigma(\varepsilon)$ by a specified error tolerance $\varepsilon_t$. This can be done by setting the machine epsilon such that 
		\begin{align}
			\epsilon_m \le \frac{\varepsilon_t}{N\sqrt{\ell(5\rme^2-4\rme)}}.\label{eq:Condition3}
		\end{align}
From Theorem~\ref{thm:normalize}, we infer that condition~\eqref{eq:Condition1} ensures
		\begin{align}
			\sigma\left(\varepsilon\right)
			\le N\epsilon_m\sqrt{5\rme^2-4\rme}.\label{eq:Condition4}
		\end{align}
Equations~\eqref{eq:Condition3} and ~\eqref{eq:Condition4} give us
	\begin{equation}
		\sigma(\varepsilon) \le \varepsilon_t
	\end{equation}
\end{proof}

As an illustration of the choice of machine epsilon, we consider the problem of approximating real and imaginary time evolution of the Hubbard model~\cite{Hubbard1964}. The Hubbard model Hamiltonian is given by  
\begin{equation}
 A_H = -t_H \sum_{\substack{\langle i,j \rangle,\\\sigma \in \{\uparrow,\downarrow\}}}( c^{\dagger}_{i,\sigma} c^{}_{j,\sigma}+ c^{\dagger}_{j,\sigma} c^{}_{i,\sigma}) + U_H \sum_{i=1}^{N} n_{i\uparrow} n_{i\downarrow}\label{eq:Hubbard}
\end{equation}
where $i$ and $j$ are labels of sites in an $\eta \times \eta$ square lattice and $\langle i,j \rangle$ represents nearest neighbor interaction. The Fermions at each site are defined by annihilation and creation operators $c_{k,\sigma}$ and $c^\dagger_{k,\sigma}$. The Hamiltonian~\eqref{eq:Hubbard} models interactions between electrons and demonstrates insulating, magnetic and superconducting effects in a solid. The model cannot be solved exactly for more than one dimension~\cite{Lieb1968} and numerical studies disagree~\cite{Hirsch1985} on the properties of the model.  

Each of the $\eta^2$ term in the nearest neighbor summation in the Hamiltonian can be represented by a $4\times 4$ matrix
\begin{equation}
A = 
\begin{pmatrix} 
0 & 0 & -t_H & -t_H	\\
0 & 0 & t_H &  t_H	\\
-t_H & t_H & U_H & 0 \\
-t_H & t_H & 0 & U_H 
\end{pmatrix}
\end{equation}
in a basis of states $\left\{| \uparrow ,\downarrow \rangle,
| \downarrow,\uparrow  \rangle,
| \uparrow \downarrow,\bullet \rangle,| \bullet, \uparrow \downarrow \rangle \right\}
$ defined over pairs of adjacent sites on the lattice. By suitable choice~\cite{Berry2006} of the TSD order $k$, we minimize the upper bound of the computational cost of simulation 
\begin{equation}
	N_{\text{exp}} = 2 \eta^4 \tau \text{e}^{2\sqrt{\ln 5 \ln (m\tau/\varepsilon_t)}}
\end{equation}
for real or imaginary simulation time $t$, $\tau = |t|\sqrt{8t_H^2+2U_H^2}$ and desired error tolerance $\varepsilon_t$. Using Eq.~\eqref{eq:Condition2}, we infer that any value $\varepsilon_m$ such that 
\begin{align}
	\epsilon_m \le \frac{\varepsilon_t\text{e}^{-2\sqrt{\ln 5 \ln (m\tau/\varepsilon_t)}}}{4 \eta^4 \tau \sqrt{(5\rme^2-4\rme)}}
\end{align}
suffices to ensure that the machine error is less than the desired error tolerance~$\varepsilon_t$.

Corollary~\ref{cor:Classical} could guide the choice of the floating-point precision for classical computation or control precision in experimental implementation depending upon the desired error tolerance. In the case of quantum computation, the norm-preserving nature of unitary operations trivially ensures condition~\eqref{eq:Condition}. We bound the error in the quantum computation of the TSD in the following corollary.
\begin{cor}
\label{cor:Quantum}
	If each $A_i$ is Hermitian and if we employ unitary operations to effect the required TSD operator exponentials, then	
	\begin{equation}
		\sigma(\varepsilon)\le  N \epsilon_m\sqrt{\ell}.
	\end{equation}
\end{cor}
\begin{proof}
In the case of unitary $U_p$ and $\tilde{U}_p$, 
		\begin{align}	
			\frac{\left\|\tilde{U}_p\right\|}{\left\|U_p\right\|} &= 1,
\end{align}
which gives us
\begin{align}
			\sigma\left(\frac{\left\|\tilde{U}_p\right\|}{\left\|U_p\right\|}\right) &= 0.		
			\end{align}
Hence, the condition~\eqref{eq:Condition} trivially holds. Additionally, the factor
\begin{equation}
 \sigma^2\left(\prod_{p=1}^N\frac{\left\| {\tilde{U}_p}\right\|}{\left\|{U_p} \right\|}\right) = 0,
\end{equation}
 which occurs in the first and third terms on the right-hand side of inequality~\eqref{eq:identity}, is zero. Consequently, we have
\begin{align}
\sigma^2\left(\frac{\left\|U_p-\tilde{U}_p\right\|}{\left\|U_p\right\|}\frac{\left\|\tilde{U}_{N-1}\dots\tilde{U}_{2}\tilde{U}_{1}\right\|}{\left\|U_{N-1}\dots U_{2}U_1\right\|}\right) \le
 \sigma^2\left(\frac{\left\|U_p-\tilde{U}_p\right\|}{\left\|U_p\right\|}\right)E^2\left(\frac{\left\|\tilde{U}_{N-1}\dots\tilde{U}_{2}\tilde{U}_{1}\right\|}{\left\|U_{N-1}\dots U_{2}U_1\right\|}\right).\label{eq:QCondition}
\end{align}
Substituting the expressions of $\sigma(\varepsilon)$~\eqref{eq:varepsilonLeq} and $\left.\sigma\left({\left\|U_p-\tilde{U}_p\right\|}\right/{\left\|U_p\right\|}\right)$~\eqref{eq:MaxError} in~\eqref{eq:QCondition}, we obtain 
	\begin{equation}\label{eq:QC}
		\sigma(\varepsilon)\le  N \epsilon_m\sqrt{\ell},
	\end{equation}
	which is linear in $N$.
\end{proof}

To summarize, we have made reasonable simplifying assumptions on the error model in the TSD. Under these assumptions, a straightforward implementation of the TSD is unstable but stability of the TSD is restored by imposing sufficiency conditions on the norm of the TSD operators. In order to attain desired error tolerance of decomposition, we provide sufficiency conditions on the floating-point precision in classical computation and control precision in experimental implementations. We have also shown that the quantum computation of the TSD has a linear (in $N$) upper bound because of the unitarity of quantum gates.

%%%%%%%%%%%%%%%%%%%%%%%%%%%%%%
%%%%%%%%%%%%%%%%%%%%%%%%%%%%%%

\section{Conclusion}
The TSD is widely used in algorithms for control and simulation of physical systems, but its stability under machine error is not proved. Here we have used a simple but reasonable model for treating errors based on the assumption of independent errors on unitary-matrix elements. This independent-error assumption allows us to conduct a rigorous mathematical analysis but does not hold in practice, either computationally or experimentally. Performing the stability analysis without the simplifying independent-error assumption is a topic for future research.

We have shown that a straightforward implementation of the TSD is unstable; i.e., the lower bound of standard error in the decomposition diverges exponentially with the number $N$ of operator exponentials in the decomposition. This exponential divergence might have, in the past, encouraged the use of lower-order product formul\ae~like the Trotter decomposition and the first order Baker-Cambell-Hausdorff formul\ae, which are sub-optimal.

We have shown that the TSD can be stable under a sufficiency condition on the norm of the TSD operators. This sufficiency condition can be satisfied by implementing normalization subroutines in classical computation, post-selection in experimental implementation. The condition is trivially satisfied in the case of quantum computation. Hence, quantum algorithms for Hamiltonian simulation as well as linear and differential equation solvers implemented even on early quantum computers without fault tolerance shall be stable against machine error. Finally, we provided sufficiency conditions on the precision of computation or control for the machine error to be bounded by a specified error tolerance.
%%%%%%%%%%%%%%%%%%%%%%%%%%%%%%
%%%%%%%%%%%%%%%%%%%%%%%%%%%%%%

\section*{Acknowledgments}
The authors acknowledge AITF, China Thousand Talent Program, CIFAR and USARO for funding. We thank Dominic W.~Berry, Andrew M.~Childs, Richard Cleve  and Nathan Wiebe for useful discussions.

\section*{References}
\bibliography{JPhysA} 

%\appendix
%\section{The distribution of norm of random matrices is a Gaussian}
%\label{App:gaussianNorm}
\end{document}

%% file: SuzukiCurved.eps_tex
%% Creator: Inkscape inkscape 0.48.2, www.inkscape.org
%% Accompanies image file 'SuzukiCurved.eps' (pdf, eps, ps)
\begingroup%
  \makeatletter%
  \providecommand\color[2][]{%
    \errmessage{(Inkscape) Color is used for the text in Inkscape, but the package 'color.sty' is not loaded}%
    \renewcommand\color[2][]{}%
  }%
  \providecommand\transparent[1]{%
    \errmessage{(Inkscape) Transparency is used (non-zero) for the text in Inkscape, but the package 'transparent.sty' is not loaded}%
    \renewcommand\transparent[1]{}%
  }%
  \providecommand\rotatebox[2]{#2}%
  \ifx\svgwidth\undefined%
    \setlength{\unitlength}{416.65bp}%
    \ifx\svgscale\undefined%
      \relax%
    \else%
      \setlength{\unitlength}{\unitlength * \real{\svgscale}}%
    \fi%
  \else%
    \setlength{\unitlength}{\svgwidth}%
  \fi%
  \global\let\svgwidth\undefined%
  \global\let\svgscale\undefined%
  \makeatother%
  \begin{picture}(1,0.88331833)%
    \put(0,0){\includegraphics[width=\unitlength]{SuzukiCurved.eps}}%
    \put(0.0,0.0){\color[rgb]{0,0,0}\makebox(0,0)[lb]{\smash{\tiny{$v(0)$}}}}%
    \put(0.82,0.69){\color[rgb]{0,0,0}\makebox(0,0)[lb]{\smash{\tiny{${v}(\lambda)$}}}}%
    \put(0.75,0.75){\color[rgb]{0,0,0}\makebox(0,0)[lb]{\smash{\tiny{${v}_{T}(\lambda)$}}}}%
    \put(0.93,0.86){\color[rgb]{0,0,0}\makebox(0,0)[lb]{\smash{\tiny{${\tilde{v}}_{T}(\lambda)$}}}}%
  \end{picture}%
\endgroup%

%% file: JPhysA.bbl
\providecommand{\newblock}{}
\begin{thebibliography}{10}
\expandafter\ifx\csname url\endcsname\relax
  \def\url#1{{\tt #1}}\fi
\expandafter\ifx\csname urlprefix\endcsname\relax\def\urlprefix{URL }\fi
\providecommand{\eprint}[2][]{\url{#2}}

\bibitem{Horn1985}
Horn R~A and Johnson C~R 1986 {\em {Matrix Analysis}\/} (Cambridge University
  Press)

\bibitem{Bandrauk1991}
Bandrauk A~D and Shen H 1991 {\em Chem.~Phys.~Lett.\/} {\bf 176} 428--432

\bibitem{Bandrauk1993}
Bandrauk A~D and Shen H 1993 {\em J.~Chem.~Phys.\/} {\bf 99} 1185

\bibitem{Vidal2004}
Vidal G 2004 {\em Phys.~Rev.~Lett.\/} {\bf 93} 40502

\bibitem{Couairon2007}
Couairon A and Mysyrowicz A 2007 {\em Phys.~Rep.\/} {\bf 441} 47--189

\bibitem{Jordan2007}
Jordan J, Or\'{u}s R, Vidal G, Verstraete F and Cirac J~I 2008 {\em
  Phys.~Rev.~Lett.\/} {\bf 101} 250602

\bibitem{Vidal2007}
Vidal G 2007 {\em Phys.~Rev.~Lett.\/} {\bf 98} 70201

\bibitem{Verstraete2008}
Verstraete F, Murg V and Cirac J~I 2008 {\em Adv.~Phys.\/} {\bf 57} 143--224

\bibitem{Banuls2009}
Banuls M~C, Hastings M~B, Verstraete F and Cirac J~I 2009 {\em
  Phys.~Rev.~Lett.\/} {\bf 102} 240603

\bibitem{Zhao2010}
Zhao J~H, Wang H~L, Li B and Zhou H~Q 2010 {\em Phys.~Rev.~E\/} {\bf 82} 61127

\bibitem{abrams1997}
Abrams D~S and Lloyd S 1997 {\em Phys.~Rev.~Lett.\/} {\bf 79} 2586

\bibitem{Lloyd1996a}
Lloyd S 1996 {\em Science\/} {\bf 273} 1073--1078

\bibitem{Harrow2009c}
Harrow A~W, Hassidim A and Lloyd S 2009 {\em Phys.~Rev.~Lett.\/} {\bf 103}
  150502 ISSN 1079-7114

\bibitem{whitfield2011}
Whitfield J~D, Biamonte J and Aspuru-Guzik A 2011 {\em Mol.~Phys.\/} {\bf 109}
  735--750

\bibitem{Jordan2012a}
Jordan S~P, Lee K~S and Preskill J 2012 {\em Science\/} {\bf 336} 1130--1133

\bibitem{Schirmer1999}
Schirmer S, Girardeau M and Leahy J 1999 {\em Phys.~Rev.~A\/} {\bf 61} 12101

\bibitem{Berry2006}
Berry D~W, Ahokas G, Cleve R and Sanders B~C 2006 {\em Commun.~Math.~Phys.\/}
  {\bf 270} 359--371

\bibitem{papageorgiou2012efficiency}
Papageorgiou A and Zhang C 2012 {\em Quantum Inf.~Process.\/} {\bf 11} 541--561

\bibitem{Blanes2009}
Blanes S, Casas F, Oteo J and Ros J 2009 {\em Phys. Rep.\/} {\bf 470} 151--238
  ISSN 03701573

\bibitem{Gilmore1974}
Gilmore R 1974 {\em J. Math. Phys.\/} {\bf 15} 2090 ISSN 00222488

\bibitem{Wilcox1967}
Wilcox R~M 1967 {\em J. Math. Phys.\/} {\bf 8} 962 ISSN 00222488

\bibitem{Magnus1954}
Magnus W 1954 {\em Communications on pure and applied mathematics\/} {\bf 7}
  649--673

\bibitem{Richtmyer1965}
Richtmyer R and Greenspan S 1965 {\em Communications on Pure and Applied
  Mathematics\/} {\bf 18} 107--108

\bibitem{Berry2013}
Berry D~W, Childs A~M, Cleve R, Kothari R and Somma R~D 2013 {\em arXiv
  preprint arXiv:1312.1414\/}

\bibitem{Berry2014}
Berry D~W, Cleve R and Gharibian S 2014 {\em Quantum Info. Comput.\/} {\bf 14}
  1--30

\bibitem{Trotter1959}
Trotter H~F 1959 {\em Proc.~Am.~Math.~Soc.\/} {\bf 10} 545--551

\bibitem{Suzuki1976}
Suzuki M 1976 {\em Commun.~Math.~Phys.\/} {\bf 51} 183--190

\bibitem{Suzuki1976a}
Suzuki M 1976 {\em Progr.~Theor.~Phys\/} {\bf 56} 1454--1469

\bibitem{Suzuki1990}
Suzuki M 1990 {\em Phys.~Lett.~A\/} {\bf 146} 319--323

\bibitem{Suzuki1991}
Suzuki M 1991 {\em J.~Math.~Phys.\/} {\bf 32} 400

\bibitem{Suzuki1992}
Suzuki M 1992 {\em Phys.~Lett.~A\/} {\bf 165} 387 -- 395

\bibitem{Suzuki1992a}
Suzuki M 1992 {\em Phys.~A Stat.~Mech.~its Appl.\/} {\bf 191} 501--515

\bibitem{Suzuki1993}
Suzuki M and Yamauchi T 1993 {\em J.~Math.~Phys.\/} {\bf 34} 4892

\bibitem{Suzuki1994}
Suzuki M 1994 {\em Commun.~Math.~Phys.\/} {\bf 163} 491--508

\bibitem{Bandrauk2013}
Bandrauk A~D and Lu H 2013 {\em JJ.~Chem.~Theory Comput.\/} {\bf 12} 1340001

\bibitem{Suzuki1995}
Suzuki M 1995 {\em Physics Letters A\/} {\bf 201} 425 -- 428

\bibitem{Tao2012}
Tao T 2012 {\em Topics in random matrix theory\/} vol 132 (American
  Mathematical Soc.)

\bibitem{Kitaev2002classical}
Kitaev A~Y, Shen A~H and Vyalyi M~N 2002 {\em {Classical and Quantum
  Computation}\/} (AMS Press)

\bibitem{Harrow2002}
Harrow A~W, Recht B and Chuang I~L 2002 {\em J.~Math.~Phys.\/} {\bf 43} 4445
  ISSN 00222488

\bibitem{Dawson2005}
Dawson C~M and Nielsen M~A 2006 {\em Quantum Inf.~Comput.\/} {\bf 6} 81--95

\bibitem{Kliuchnikov2013}
Kliuchnikov V, Maslov D and Mosca M 2013 {\em Quantum Info.~Comput.\/} {\bf 13}
  607--630

\bibitem{Johnson1997}
Johnson N~L, Kotz S and Balakrishnan N 1997 {\em {Continuous Univariate
  Distributions}\/} vol~1 (John Wiley \& Sons)

\bibitem{Leone1961}
Leone F~C, Nelson L~S and Nottingham R~B 1961 {\em Technometrics\/} {\bf 3}
  543--550

\bibitem{Bai1999}
Bai Z 1999 {\em Stat.~Sin.\/} {\bf 9} 611--677

\bibitem{Ku1966}
Ku H~H 1966 {\em J.~Res.~Natl.~Bur.~Stand.~- C\/} {\bf 70C} 75--79

\bibitem{Hubbard1964}
Hubbard J 1964 {\em Proc. R. Soc. A Math. Phys. Eng. Sci.\/} {\bf 281} 401--419
  ISSN 1364-5021

\bibitem{Lieb1968}
Lieb E and Wu F 1968 {\em Phys. Rev. Lett.\/} {\bf 20} 1445--1448

\bibitem{Hirsch1985}
Hirsch J 1985 {\em Phys. Rev. B\/} {\bf 31} 4403--4419 ISSN 0163-1829

\end{thebibliography}
